\documentclass{article}

\usepackage{arxiv}

\usepackage[utf8]{inputenc} % allow utf-8 input
\usepackage[T1]{fontenc}    % use 8-bit T1 fonts
\usepackage{hyperref}       % hyperlinks
\usepackage{url}            % simple URL typesetting
\usepackage{booktabs}       % professional-quality tables
\usepackage{amsfonts}       % blackboard math symbols
\usepackage{nicefrac}       % compact symbols for 1/2, etc.
\usepackage{microtype}      % microtypography
\usepackage{lipsum}		% Can be removed after putting your text content
\usepackage{graphicx}
\newtheorem{theorem}{Theorem}

\newtheorem{corollary}[theorem]{Corollary}

\newtheorem{definition}[theorem]{Definition}

\newtheorem{proposition}[theorem]{Proposition}

\newenvironment{proof}[1][Proof]{\noindent\textbf{#1.} }{\ \rule{0.5em}{0.5em}}

\title{Sampling spatial \ Structures in geostatistical framework}

\author{\href{https://orcid.org/0000-0000-0000-0000}{\includegraphics[scale=0.06]{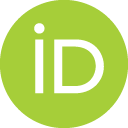}\hspace{1mm}Fabrice Ouoba} \\
	Université de Fada N'Gourma\\
	BP: 54 Fada N'Gourma, Burkina Faso\\
	\texttt{didifab@yahoo.fr} \\
	\And	
	\href{https://orcid.org/0000-0000-0000-0000}{\includegraphics[scale=0.06]{orcid.png}\hspace{1mm}Diakarya Barro} \\
	Université Thomas Sankara \\
	12 BP: 417 Ouagadougou, Burkina Faso\\
	\texttt{dbarro2@gmail.com} \\
	\And
    \href{https://orcid.org/0000-0000-0000-0000}{\includegraphics[scale=0.06]{orcid.png}\hspace{1mm}Hay Yoba Talkibing} \\
	Université Joseph KI ZERBO\\
	03 BP: 7021 Ouagadougou, Burkina Faso\\
	\texttt{talkibingfils@yahoo.fr} 
}

\begin{document}
\maketitle

\begin{abstract}
Extreme values geostatistics make it possible to model the asymptotic behaviors of random phenomena which depends on space or time
parameters. In this paper, we propose new models of the extremal coefficient within a spatial stationary fields underlied by multivariate
copulas. Some models of extensions of the extremogram and the cross-extremogram are constructed in a spatial framework. Moreover, both
these two geostatistcal tools are modeled using the extremal variogram which characterizes the asymptotic stochastic behavior of the phenomena.
\end{abstract}


\section{Introduction}
Geostatistics provide many tools for statistical analysis of spatial or
spatio temporal datasets. This branch of statistics was developed originally
in the years $1930$ by a pionering work of George Matheron \cite{ato1}
to predict the probability distributions of more grades for mining operations.
Since, it became a subdomain of statistics based on the notion of random
fields including petroleum geology, hydrogeology, geochemistry,
geometallurgy, geography, forestry, environmental control, landscape
ecology, soil science and agriculture.

In spatial statistical analysis, the variograms and the covariance functions
are technical tools used in describing how the spatial continuity
changes with a given separating distance between two pair of stations. So, the classical variogram provides a framework for
modelling and predicting the variability of a given the stochastic spatial
process.

The family of copulas provides a natural way to construct multivariate
distribtutions whose marginals are uniform and not necessarily exchangeable.
Let $X=\left( X_{1},...,X_{n}\right) $ be a random vector with multivariate
continuous distribution function (c.d.f.) H and c.d.f marginal $%
H_{1},...,H_{n}.$ The copula of X (of the c.d.f. H respectively) is the multivariate c.d.f. C of the random vector $U=\left[ H_{1}(x_{1}),\dots,H_{n}(x_{n})\right] $. Due to the continuity of $\left\{ H_{i},1\leq i\leq
n\right\} ,$ each component of U is standard uniformly distributed, i.e., $U_{i}\sim U(0,1)$ for $i=1,\dots,n.$

Particularly, every n-copula must satisfy the n-increasing property \cite{frahm}. That means that, for any rectangle $B=\left[ a,b\right] ^{n}\subseteq \mathbb{R}^{n},$ the B-volume $C_{B}$ of C is positive, that is, 
\begin{equation}
C_{B}=\int_{B}dC\left( u\right) =\sum_{i_{1}=1}^{2}\dots\sum_{i_{n}=1}^{2}\left( -1\right)^{i_{1}+\dots +i_{n}}C\left( u_{1i_{1}},\dots,u_{1i_{n}}\right) \geq 0.
\end{equation}

In multivariate copulas analysis following canonical
parameterization of H (see \cite{ot1}) allows the\ use of copulas in
stochastic analysis under the so called Sklar theorem $\left[ 9\right] $\ or
Nelsen $\left[ 12\right] $\ 
\begin{equation}
\begin{tabular}{l}
$C(u_{1},...,u_{n})=H[H_{1}^{-1}(u_{1}),...,H_{n}^{-1}(u_{n})]$.
\end{tabular}
\end{equation}

H$^{-1}$ being the generalized inverse such as $H^{-1}\left(
x\right) =\inf \left\{ t\in \left[ 0,1\right] ,F\left( t\right) \leq
x\right\} $\textit{.\medskip }

While modeling the main geostatistical tools Ouoba et al. \cite{pater}
\ have provided the copula-based variogram, correlogram and madogam and they pointed out that these tools do not take into account the
extreme data observed in the different observation sites. However, the
copula function makes it possible to model the extreme data and make it
possible to detect any nonlinear link between different observation sites.
It is therefore necessary to express the variogram and the covariogram
according to the copula in order to be able to model the spatial structuring
even if our giving includes extremes and to be able to detect the presence
of some nonlinear dependence.\newline
The variogram 
\begin{eqnarray*}
\vartheta (s_{i},s_{j})=Var(Z(s_{i}))+Var(Z(s_{j}))-2\hat{c}(s_{i},s_{j}),
\end{eqnarray*}%
the covariogram $\hat{c}(s_{i},s_{j})$ and the copula function are linked by
the relation: 
\begin{eqnarray*}
\vartheta (s_{i},s_{j})=\sigma _{Z}^{2}(s_{i})+\sigma
_{Z}^{2}(s_{j})-2\int_{0}^{1}%
\int_{0}^{1}F_{Z}^{-1}(u)F_{Z}^{-1}(v)c(u,v)dudv-2m_{i}m_{j}
\end{eqnarray*}%
\begin{eqnarray*}
\hat{c}(s_{i},s_{j})=\int_{0}^{1}%
\int_{0}^{1}F_{Z}^{-1}(u)F_{Z}^{-1}(v)c(u,v)dudv-m_{i}m_{j}
\end{eqnarray*}%
where $m_{i}$ and $m_{j}$ the respective averages of $Z(s_{i})$ and $Z(s_{j})
$; $c(u,v)$ the copula density function attached to $Z(s_{i})$ and $Z(s_{j})$%
.\\
\\
The major contribution of this article is to provide tools to model the dependence of extremes in the mining context. In section 2 we develop the tools needed to achieve our goals. Our main results are given in section 3, where we propose new models of the extremal coefficient in a spatial stationary field using multivariate copulas and the extensions of the extremogram and the crossed extremogram in a spatial framework using the extremal variogram which characterizes the asymptotic stochastic behavior of phenomena.
\section{Back Ground}
In this section we collect the necessary definitions and usefull properties
on extremal dependence coefficient and tail dependence. So an overview of
spatial framework and copulas functions is given as well as some statements
of multivariate tail dependence coefficients.\bigskip

Multivariate extreme values (MEV) theory is often presented in the framework
of coordinatewise maxima, so the importance of distinction diminishes.
Towards a multivariate analogue of Fisher-Tippett we are looking for some
sort of multivariate limit distribution for conveniently normalized vectors
of multivariate maxima. For an arbitrary index of set T denoting generally a
space of time, a random vector $Y_{t}=\left\{ Y_{j}\left( t\right) ;1\leq
j\leq m,t\in T\right\} $\ in $%
%TCIMACRO{\U{211d} }%
%BeginExpansion
\mathbb{R}
%EndExpansion
^{m}$\ is said to be max-stable if, for all $n\in 
%TCIMACRO{\U{2115} }%
%BeginExpansion
\mathbb{N}
%EndExpansion
,$\ every $Y_{j}\left( t\right) =(Y_{j}^{(1)}\left( t\right) ;\ldots
;Y_{j}^{(n)}\left( t\right) )$\ is a n-dimensionnal max-stable vector, that
is, there exists suitable and time-varying non-random sequences $\left\{
a_{n}\left( t\right) >0\right\}~~ and ~~\left\{ b_{n}\left( t\right) \in \mathbb{R}^{d}\right\}~~ such~~ as 
~~$
\begin{eqnarray}
\frac{1}{a_{n}(t)}\left[ M_{n}(t) - b_{n}(t) \right] \stackrel{f.d.d}{\longrightarrow} X(t) ;t\in T,
\end{eqnarray}

where $\stackrel{f.d.d}{\longrightarrow}$ denotes the convergence for the
finite-dimensional distributions while $M_{n}\left( t\right) =\max_{1\leq i \leq n} \left( X_{i}\left(
t\right) \right);t\in T$ being the component-wise maxima of the time-variying vector $X\left( t\right) $.\medskip

Like in the non-spatial analysis, several canonical representations of
max-stable processes have been suggested in spatial extreme values context.
\ In the same vain, Barro et al. (see \cite{ato3} ) have propose the
following result \ allows us to characterize the general form of the
one-dimensional marginal of the max-stable ST process $\left\{ Y_{t}\right\} $ where $Y_{t}\left( x\right) =Y\left( x_{t}\right) $; $x_{t}\in \chi
_{_{D}}\times T\subset R^{3}$. $\left( Y_{j}^{s}\left( t\right) \right)
;j\geq 0;t\in T;s\in S$ such that for each fixed couple $(t,s)$, the
sequence is independent and identically distributed according to a joint
cumulative function $G_{t}^{s}$. Under the assumption that this function is
max-stable, every univariate margins $G_{t,i}^{s}$\ \ lies its own domain of
attraction and is expressed by on the space of interest $S_{\xi
_{i},t,s}^{+}=\left\{ z\in 
%TCIMACRO{\U{211d} }%
%BeginExpansion
\mathbb{R}
%EndExpansion
;\sigma _{i,t,s}+\xi _{i,t,s}\left( y_{t,i}^{s}-\mu _{i,t,s}\right) >0;1\leq
i\leq n\right\} $\ by 
\begin{eqnarray}
G_{i}\left( y_{i}\left( s\right) \right) =\left\{ 
\begin{tabular}{l}
$\exp \left\{ -\left[ 1+\xi _{i}\left( s\right) \left( \frac{y_{i}\left(
s\right) -\mu _{i}\left( s\right) }{\sigma _{i}\left( s\right) }\right) %
\right] ^{\frac{-1}{\xi _{i}\left( s_{i}\right) }}\right\} 
\hspace*{0.5cm} $ \textit{if }$% 
\xi _{i}\left( s\right) \not=0$ \\ 
\\
\\
$\exp \left\{ -\exp \left\{ -\left( \frac{y_{i}\left( s\right) -\mu
_{i}\left( s\right) }{\sigma _{i}\left( s\right) }\right) \right\} \right\} $
\hspace*{1.7cm} \textit{if }$\xi _{i}\left( s\right) =0$%
\end{tabular}%
\right. ;
\label{itp}
\end{eqnarray}%
and for all site s, the parameters $\left\{ \mu _{i,t,s}\in 
%TCIMACRO{\U{211d} }%
%BeginExpansion
\mathbb{R}
%EndExpansion
\right\} $, $\left\{ \sigma _{i,t,s}>0\right\} $\ and $\left\{ \xi
_{i,t,s}\in 
%TCIMACRO{\U{211d} }%
%BeginExpansion
\mathbb{R}
%EndExpansion
\right\} $\ are referred to as the location, the scale and the shape
parameters respectively. Particularly, the different values of $\xi
_{i}\left( s\right) \in \mathbb{R}$\ allows \ref{itp} to be a spatial EV model, that is, to belong either to Frechet family, the Weibull one or Gumbel one.\medskip

In multivariate case if the one-dimensional margins of F are unit-Fréchet distributed let M be a non-empty subset of $N=\left\{ 1,...,n\right\} $
and $c_{_{M}}$ the n-dimensional vector of which the jth coordinate is one
or zero according to $j\in M$ or $j\not\in M$. Then, the multivariate, $\theta _{M}$ is defined on the n-dimensional unit simplex, S$_{n}=\left\{
\left( t_{1}...,t_{n}\right) \in \left[ 0,1\right] ^{n},\sum_{i=1}^{n}t_{i}\leq 1\right\} ,$ such as,
\begin{equation}
\theta _{M}=V(c_{_{M}})=\int_{S_{n}}\max_{j\in M} \left( \frac{w_{j}}{\left\Vert w\right\Vert _{1}}\right) dH\left( w_{j}\right) ,
\end{equation}%
where H is a finite non-negative measure of probability and $\left\Vert
.\right\Vert _{1},$ the 1-norm, see \cite{ot1}, \cite{ot4}.
Particularly 
\begin{equation}
P\left[ F_{1}\left( x_{1}\right) \leq p,...,F_{n}\left( x_{n}\right) \leq p \right] =p^{\theta }~~ for ~~all ~~0<p<1.
\end{equation}

In spatial study, a natural way to measure dependence among spatial maxima
stems from considering the distribution of the largest value that might be
observed on domain of study. 

Our main results are summaried by the following sections.

\section{Spatial max-stability within geostatiscal framework}

The extremal coefficient is the natural dependence measures for extreme
value models which provides the magnitude of the asymptotic dependence of a
random field at two points of the domain.

\subsection{Context and definitions}
The context in this study, is a mining ressource models one.
\subsubsection{Problematic and variable}

We consider a mining geographic area for example. By considering a subdivision, a paving of the domain. We consider a reference grade $\beta_{O}$ for the given ore. The domain is thus divided into two subdomains depending on whether the content of the locality is lower or higher than this reference.

Working hypotheses:

\begin{description}
\item $\bullet $ The content depends on the locality and the depth h, that is to say $s=s\left( h\right) .$

\item[$\bullet $] The variable Y representing the content is linked to the locality $Y=Y\left( s\right) $ and therefore $Y=Y\left( s\left( h\right)
\right).$ So we build a stochastic process?
\end{description}

In this study, let $\left\{ Y_{s},s\in S\right\} $ a spatial stochastic
process defined on a geographical domain $S=\left\{ s_{1},\dots
,s_{n}\right\} $ \ where $\left\{ Y_{s},s\in S\right\} $ denotes the
contents of a metal in a mining site.\newline

Let n,k be naturel numbers such that $\left\{ n\geq 2;~~1\leq k\leq
n\right\} $ and let $N_{k}$ be a given subset of k elements of $N=\left\{
1,\dots ,n\right\} ,$ the set of the first n natural numbers.

\begin{figure}[htbp!!]
\begin{center}
\includegraphics[scale=0.7]{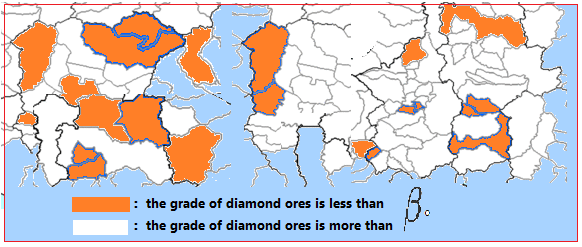} 
\caption{Domain of mining study.}
\end{center}
\end{figure}

Even in spatial stochastic context, three possible distributions can
describe the asymptotic behavior of conveniently normalized extremal
distributions at a given geographical locality s. These distributions are
instead described by a class of dependence models. Specially in a spatial
framework, let $D_{N}=\left\{ s_{1},...,s_{N}\right\} \subset \mathbb{R}^{2}$\textit{\ be the set of locations (geographical ereas, mines
localities, ...), sampled over a }$\left[ 0,\frac{1}{s}\right] \times \left[
0,\frac{1}{m}\right] $\textit{\ rectangle }$\left( m\in \mathbb{N}
%EndExpansion
\right) $\textit{, where the phenomenas are observed. Let Y a variable of
interest, observed at given site s and date t.}

\textit{So, the relation (2.2) provides, for all }$x_{t}=\left(
x_{t}^{\left( 1\right) };...;x_{t}^{\left( m\right) }\right) $\textit{\ in }$\mathbb{R}^{m}\times T$\textit{\ \ the relation}%
\begin{equation}
C_{s(t)}(u_{1};...;u_{m})=F_{t}^{\check{s}}\left[\left( F_{t}^{\check{s}_{1}}(u_{1})\right) ^{-1};...;\left( F_{t}^{\check{s}%
_{m}}(u_{m})\right) ^{-1}\right] .
\label{ifo}
\end{equation}

\textit{Note that, for all }$m\in \mathbb{N}
$ and for all geographical locality s, the spatio-temporal unit
simplex of $\mathbb{R}^{(m-1)}$ is given, under the notational by 
\begin{equation}
\Delta _{t,m}^{\check{s}}=\left\{ \lambda _{t}^{\check{s}}=\left( \lambda
_{1}^{\check{s}_{1}};...;\lambda _{t}^{\check{s}_{m}}\right) \in\mathbb{R}_{+}^{m};\left\Vert \lambda _{t}^{\check{s}}\right\Vert
=\sum_{i=1}^{m}\lambda _{t}^{\check{s}_{i}}=1\right\} .
\end{equation}

\subsubsection{Spatial max-stability}

\begin{definition}
Let $Y=\left\{ Y_{s(h)},s\in S,t\in T\right\}$ be a spatial process with
parametric joint distribution $H_{t}^{\check{s}}.$\textit{\ The following statements are
satisfied \ a sufficient condition for the process }$H_{t}^{\check{s}}$%
\textit{\ to be a ST-MEV distribution\ is that there exists two
spatio-temporal non-random sequences }$\left\{ \alpha _{n}^{\check{s}}\left(
t\right) >0\right\} $\textit{\ and }$\left\{ \beta _{n}^{\check{s}}\left(
t\right) \in \mathbb{R}\right\} $ such that
\begin{eqnarray}
\lim_{n\uparrow \infty} P\left( \frac{M_{t}^{\check{s}}-\beta _{n}^{\check{s}}\left( t\right) }{\alpha _{n}^{\check{s}}\left(
t\right) }\leq y_{t}^{\check{s}}\right) =\left( H_{1}\left( y_{t}^{\check{s}_{1}}{}\right) ,\dots,H_{m}\left( y_{t}^{\check{s}_{m}}{}\right) \right) .
\label{efoo}
\end{eqnarray}
\textit{where }$M_{t}^{\check{s}\left( i\right) }$\textit{\ is univariate
margins of the spatio-temporal componentwise vector of maxima. }
\end{definition}

As a corrolary of the above definition, (\ref{ifo})
provides, for all $x_{t}=\left( x_{t}^{\left( 1\right) };...;x_{t}^{\left(
m\right) }\right) $\textit{\ in }$\mathbb{R}^{m}\times T$\textit{\ \ the relation}%
\begin{equation}
C_{s(t)}(u_{1}^{t};...;u_{m}^{t})=H_{t}^{\check{s}}\left[\left( H_{1}^{\check{s}_{1}}(u_{1}^{t})\right) ;...;\left( H_{m}^{\check{s}%
_{m}}(u_{m}^{t})\right) \right],
\label{efo}
\end{equation}
where $u_{i}^{t}=H_{i}^{-1}\left(s_{i}(t)\right).$\\
\textit{Note that, for all }$m\in 
%TCIMACRO{\U{2115} }%
%BeginExpansion
\mathbb{N}
%EndExpansion
$\textit{\ and for all geographical locality s, the spatio-temporal unit
simplex of }$%
%TCIMACRO{\U{211d} }%
%BeginExpansion
\mathbb{R}
%EndExpansion
^{(m-1}$\textit{\ is given, under the notational by }%
\begin{equation}
\Delta _{t,m}^{\check{s}}=\left\{ \lambda _{t}^{\check{s}}=\left( \lambda
_{1}^{\check{s}_{1}};...;\lambda _{t}^{\check{s}_{m}}\right) \in 
%TCIMACRO{\U{211d} }%
%BeginExpansion
\mathbb{R}
%EndExpansion
_{+}^{m};\left\Vert \lambda _{t}^{\check{s}}\right\Vert
=\sum_{i=1}^{m}\lambda _{t}^{\check{s}_{i}}=1\right\} .
\end{equation}

\subsubsection{Domain Spatially discordant}
In this sub-section, we consider a geographical domain D made up of several sites that we partition into two sub-domains depending on whether the sites in the domain have a mineral content higher than a reference value or not.
\begin{definition}
(Domain Spatially discordant)We define $N_{k}$-partition of a random vector $X=\left\{ X_{1},\dots ,X_{n},~~n\geq 2\right\} $ (or the partition of X in
the direction of $N_{k}$) by the pairwise vector $\stackrel{\sim }{X}=(\stackrel{\sim }{X}_{N_{k}},\stackrel{\sim }{X}_{\overline{N}_{k}})$ as:

$\bullet \stackrel{\sim }{X}=(X_{N_{k},1},\dots ,X_{N_{k},k})$ is the
k-dimensional marginal vector of X whose component indexes are ordered in
the subset $N_{k}$.

$\bullet \stackrel{\sim }{X}=(X_{\overline{N}_{k},1},\dots ,X_{\overline{N}_{k},n-k})$ is the $(n-k)$ dimensional marginal vector of X whose component
indexes are ordered in $\overline{N}_{k}=C_{N}^{N_{k}}$, the complementary
of $N_{k}$in N.
\end{definition}

Similarly, every realisation $x=(x_{1},\dots ,x_{n})$ of X can be decomposed
into two parts 
\begin{eqnarray*}
x=(\stackrel{\sim }{x}_{N_{k}},\stackrel{\sim }{x}_{\overline{N}_{k}})~~where~~\stackrel{\sim }{x}_{N_{k}}=(x_{N_{k},1},\dots
,x_{N_{k},k})~~and~~\stackrel{\sim }{x}=(x_{\overline{N}_{k},1},\dots ,x_{\overline{N}_{k},k})
\end{eqnarray*}%
are, respectively realizations of vectors $\stackrel{\sim }{X}_{N_{k}}$ et $%
\stackrel{\sim }{X}_{\overline{N}_{k}}$. If H, $H_{N_{k}}$ and $H_{\overline{N}_{k}}$ denote the distribution functions of the random vectors X, $%
\stackrel{\sim }{X}_{N_{k}}$ and $stackrel{\sim }{X}_{\overline{N}_{k}}$,
then for all realization $x=(x_{1},\dots ,x_{n})$ of X we have 
\begin{eqnarray*}
H_{N_{k}}(\stackrel{\sim }{x}_{N_{k}}) &=&\displaystyle\lim_{\stackrel{\sim }{x}_{\overline{N}_{k}}\rightarrow \stackrel{\sim }{x}_{\overline{N}%
_{k}}^{\ast }}H(x)~and~H_{\overline{N}_{k}}(\stackrel{\sim }{x}_{\overline{N}_{k}})=\displaystyle\lim_{\stackrel{\sim }{x}_{N_{k}}\rightarrow \stackrel{\sim }{x}_{N_{k}}^{\ast }}H(x) \\
~where~\stackrel{\sim }{x}_{N_{k}}^{\ast } &=&(x_{N_{k},1}^{\ast },\dots
,x_{N_{k},k}^{\ast })~and~\stackrel{\sim }{x}^{\ast }=(x_{\overline{N}%
_{k},1}^{\ast },\dots ,x_{\overline{N}_{k},k}^{\ast })
\end{eqnarray*}%
are the upper endpoints of the functions $H_{N_{k}}$ and $\stackrel{\sim }{X}%
_{\overline{N}_{k}}$.

\begin{definition}
Given a $N_{k}-$partition $\stackrel{\sim }{X}=\left\{ \left( X_{N_{k}},
\stackrel{\sim }{X}_{\overline{N}_{k}}\right) ,1\leq k\leq n\right\} $ of $X=\left\{ X_{1},\dots ,X_{n}\right\} $ we define the upper $N_{k}-$%
discordance degree of X as the conditional probability given for all $%
x=(x_{1},\dots ,x_{n})\in \mathbb{R}_{n}$ by $\delta _{N_{k}}^{+}(x)=P(%
\stackrel{\sim }{X}_{N_{k}}>\stackrel{\sim }{x}_{N_{k}}/\stackrel{\sim }{X}_{%
\overline{N}_{k}}\leq \stackrel{\sim }{x}_{\overline{N}_{k}})$.
Similarly, the lower $N_{k}-$discordance degree of X is defined, for all $%
x=(x_{1},\dots ,x_{n})\in \mathbb{R}_{n}$ by $\delta _{N_{k}}^{-}(x)=P(%
\stackrel{\sim }{X}_{N_{k}}\leq \stackrel{\sim }{x}_{N_{k}}/stackrel{\sim }{X}%
_{\overline{N}_{k}}>\stackrel{\sim }{x}_{\overline{N}_{k}})$.\newline
\end{definition}

\subsubsection{Spatial discordance rate}
In the modeling of spatial extremes, the calculation of quantiles is very important.
The following definition characterizes the probability that one of the
margins $\stackrel{\sim }{X}_{N_{k}}$ and $\stackrel{\sim }{X}_{\overline{N}%
_{k}}$ exceeds $1/2$, while the values taken by the other are less than $%
1/2. $

\begin{definition}
Given the distribution H of a multivariate random
vector $X=\left\{ X_{1},\dots ,X_{n}),n\geq 2\right\} $ with univariate
margins $H_{i},1\leq i\leq n$ we define the upper $N_{k}-$median discordance
degree of H by the real number denoted by $\delta _{N_{k},H}^{+}$ such as: $%
\delta _{N_{k},H}^{+}=\delta _{N_{k}}^{+}\left[ (H_{1}^{-1}(\frac{1}{2}%
),\dots ,H_{n}^{-1}(\frac{1}{2}))\right] $ where $H_{i}^{-1}$ is quantile
function of $H_{i}.$ Similarly, the lower $N_{k}-$median discordance degree
of H is defined by $\delta _{N_{k},H}^{-}=\delta _{N_{k}}^{-}\left[
(H_{1}^{-1}(\frac{1}{2}),\dots ,H_{n}^{-1}(\frac{1}{2}))\right] $
\end{definition}

\subsection{Spatial MDA and inferential properties}

Spatial max-stable processes generalize the Multivariate Extreme
Value (MEV) laws to the spatial context and hold information on the spatial
dependence structure. Specifically a constructive definition is given as
follows.

\begin{definition}
\textit{(see \cite{ato3}) Let S be a spatial domain. We say that the process }$%
\left\{ Y(s),s\in \mathcal{S}\right\} $\textit{\ is max-stable if all the
marginal distributions are max-stable, that is to say exists for all }$n$%
\textit{\ two suites of continuous functions }$\left\{ \alpha
_{n}(s)\right\} >0$\textit{\ and }$\left\{ \beta _{n}(s)\right\} $\textit{\
such as: }%
\begin{eqnarray*}
\lim_{n\rightarrow +\infty }\left\{\frac{\max_{i=1}^{n}X(s_{i})-\beta
_{n}(s)}{\alpha _{n}(s)}\right\} =\left\{ Y(s),s\in \right\} ,
\end{eqnarray*}
\end{definition}

with $X(s_{i})$ independent and identically distributed
copies of $X$, a stochastic process representing for example a meteorological parameter. Without loss of generality and to consider only
the spatial dependence of $\left\{ Y(s),s\in \right\} $, it is more convenient to transform $Y$ into a simple max-stable process, ie
with Frechet margins unit (ie $GEV\left( 1,1,1\right) $) for all $s\in S$ via the following transformation: 
\begin{eqnarray*}
Y_{F}(s)=\frac{-1}{\log \left\{ G_{\mu (s),\sigma (s),\xi (s)}\left(
Y(s)\right) \right\} }.
\end{eqnarray*}

The study of extreme value theory have been extended both to spatial and
multivariate contexts these last years. This section gives the relationship
between the extremal coefficient via copula.

\subsection{Stability spatial marginal}

\begin{theorem}
Let G be a spatially max-stable multivariate distribution.. \ \textit{Then,
under the condition of the max-stability of }$\left\{ Y_{{}}\right\} $%
\textit{, the distributions underlying the marginal processes }$\left\{
Y_{t_{A}}\mathit{\ }\right\} $ \textit{and }$\left\{ Y_{t_{\bar{A}}}\mathit{%
\ }\right\} $\textit{\ lies \ respectively in the MDA of two parametric MEV
models G}$_{A}$\textit{\ and G}$_{\bar{A}}$\textit{. Moreover the
distributions G}$_{A}$\textit{\ and G}$_{\bar{A}}$\textit{\ are marginal
distributions of G.}
\end{theorem}

\begin{proof}
\textit{\ Let }$\left\{ \alpha _{n}>0\right\} $\textit{\ and }$\left\{ \beta
_{n}\in 
%TCIMACRO{\U{211d} }%
%BeginExpansion
\mathbb{R}
%EndExpansion
\right\} $\textit{\ be the non-random \ normalizing sequences of H. Then,
their corresponding space and time extensions }$\left\{ \alpha _{n}^{\check{s%
}}\left( t\right) >0\right\} $\textit{\ and }$\left\{ \beta _{n}^{\check{s}%
}\left( t\right) \in 
%TCIMACRO{\U{211d} }%
%BeginExpansion
\mathbb{R}
%EndExpansion
\right\} $\textit{\ are defined on the set, }$%
%TCIMACRO{\U{2115} }%
%BeginExpansion
\mathbb{N}
%EndExpansion
^{\ast }\times S\times T,$ such that
\begin{eqnarray*}
\lim_{n\uparrow \infty} P\left( \frac{M_{t}^{\check{s}}-\beta _{n}^{\check{s}}\left( t\right) }{\alpha_{n}^{\check{s}}\left(t\right) }\leq y_{t}^{\check{s}}\right) =\lim_{n\uparrow\infty} P\left[ \cap_{i=1}^{n}\left( \frac{
M_{t}^{\check{s}_{i}}-\beta _{i}^{\check{s}_{i}}\left( t\right) }{\alpha
_{i}^{\check{s}_{i}}\left( t\right) }\leq y_{t}^{\check{s}_{i}}\right) %
\right]
\end{eqnarray*}%
\textit{Then, }%
\begin{eqnarray*}
\lim_{n\rightarrow\infty} P\left( \frac{M_{t}^{\check{s}}-\beta _{n}^{\check{s}}\left( t\right) }{\alpha_{n}^{\check{s}}\left(
t\right)}\leq y_{t}^{\check{s}}\right) =\lim_{n\rightarrow\infty} P\left[ \cap_{i=1}^{n}\left( Y_{t}^{\check{s}_{i}}\leq \alpha _{i}^{\check{s}}\left( t\right) y_{i}^{\check{s}}\left( t\right) +\beta _{i}^{\check{s}}\left( t\right) \right) \right] .
\end{eqnarray*}

That is equivalent, due to independence, to 
\begin{eqnarray*}
\lim_{n\rightarrow\infty} P\left( \frac{M_{t}^{\check{s}}-\beta _{n}^{\check{s}}\left( t\right) }{\alpha _{n}^{\check{s}}\left(
t\right) }\leq y_{t}^{\check{s}}\right) =\lim_{n\rightarrow\infty} \left( \prod_{i=1}^{n}P\left[ \left( X_{i}^{\check{s}}\leq \alpha _{i}^{\check{s}}\left( t\right) y_{i}^{\check{s}}\left( t\right) +\beta _{i}^{\check{s}}\left( t\right) \right) \right]
\right) .
\end{eqnarray*}

So, there exists a max-stable distribution G whose max-domain of attraction contains the MEV H. Then,
\begin{eqnarray*}
\lim_{n\rightarrow\infty} P\left( \frac{M_{t}^{\check{s}}-\beta _{n}^{\check{s}}\left( t\right) }{\alpha _{n}^{\check{s}}\left(
t\right) }\leq y_{t}^{\check{s}}\right) =\lim_{n\rightarrow\infty} \left[ G\left( \alpha _{i}^{\check{s}}y_{i}{}^{\check{s}}\left( t\right) +\beta _{i}^{\check{s}}\left( t\right) \right) ,\dots \alpha
_{i}^{\check{s}}\left( t\right) y^{\check{s}}\left( t\right) +\beta _{i}^{\check{s}}\right] ^{n}.
\end{eqnarray*}

Finally, since the distribution G is max-stable 
\begin{eqnarray*}
\lim_{n\rightarrow\infty} P\left( \frac{M_{t}^{\check{s}}-\beta _{n}^{\check{s}}\left( t\right) }{\alpha _{n}^{\check{s}}\left(
t\right) }\leq y_{t}^{\check{s}}\right) =\left( H_{1}\left( y_{t}^{\check{s}\left( 1\right) }{}\right) ,\dots,H_{n}\left( y_{t}^{\check{s}\left( n\right)
}{}\right) \right) .
\end{eqnarray*}

The process $\left\{ Y_{t}\right\} $ is max-stable by assumption, so
Corollary 4 (see \cite{ato3}) implies that the underlying distribution lies in the MDA\ of \ a
parametric extreme values model G.\ Equivalently there exist the normalizing
sequences as in $(3.3)$ such as, for all $x_{t}\in \chi \times T,$ 
\begin{equation}
\lim_{n\rightarrow\infty} P\left( \cap_{i=1}^{n}\left\{ \frac{M_{i}\left( x_{t}\right) -\mu
_{i}\left( x_{t}\right) }{\sigma _{i}\left( x_{t}\right) }\leq y_{i}\right\}
\right) =G\left( y_{1}\left( x_{t}\right) ,\dots,y_{n}\left( x_{t}\right)
\right) .
\end{equation}

Setting $x_{t}=(x_{t_{A}},x_{t_{\bar{A}}})\in \chi \times T,$ the marginal
distribution $G_{A}$ of G defined on the sub-domain $\chi _{A}$ is obtained
asymptotically by
\begin{eqnarray*}
G_{A}(y\left( \mathit{\ }x_{t_{A}}\right) ) &=&\lim_{x_{t_{\bar{A}}}\rightarrow x_{t_{\bar{A}}}^{\ast }} G(y\left(
x_{t}\right) ) \\
&=&\lim_{x_{t_{\bar{A}}}\rightarrow x_{t_{\bar{A}}}^{\ast }}\left[ \lim_{n\rightarrow\infty} P\left( \cap_{i=1}^{n}\left\{ \frac{M_{i}\left(
x_{t}\right) -\mu _{i}\left( x_{t}\right) }{\sigma _{i}\left( x_{t}\right) }\leq y_{i}\right\} \right) \right]
\end{eqnarray*}
\textrm{\ where }$x_{t_{\bar{A}}}^{\ast }$ \textrm{is the right endpoint of
the distribution }$G_{\bar{A}}.$ Then, it follows that%
\begin{eqnarray*}
G_{A}(y\left( x_{t_{A}}\right) ) &=&\lim_{n\rightarrow\infty} \left[ \lim_{x_{t_{\bar{A}}}\rightarrow x_{t_{\bar{A}}}^{\ast }} P\left( \cap_{i=1}^{n}\left\{ \frac{M_{i}\left( x_{t}\right) -\mu _{i}\left( x_{t}\right) }{\sigma
_{i}\left( x_{t}\right) }\leq y_{i}\right\} \right) \right] \\
&=& \lim_{n_{A}\longrightarrow +\infty} \left[ P\left( \cap_{i_{A}=1}^{n_{A}} \left\{ \frac{M_{i_{A}}\left(
x_{t_{A}}\right) -\mu _{i_{A}}\left( x_{t_{A}}\right) }{\sigma
_{i_{A}}\left( x_{t_{A}}\right) }\leq y\left( x_{t_{A}}\right) \right\}
\right) \right]
\end{eqnarray*}

where the index $i_{A}$ is such as that $x_{t_{A}}\in \chi _{A}$.
Therefore, there exist marginal ST normalizing sequences $\left( \sigma
_{n_{A}}\left( x_{t_{A}}\right) ,\mu _{n_{A}}\left( x_{t_{A}}\right) \right)
\in \mathbb{R}^{+}\times \mathbb{R}$ such that\textrm{, the corresponding marginal component-wise maxima\
converge to }$G_{A}$\textrm{\ according\ equality }$\left( 3.2\right) .$ 
\textrm{Finally, the underlying distribution of the ST marginal process }$%
Y_{A}$ lies in the MDA of the n$_{A}$-dimensional parametric MEV
distribution $G_{A}$, $n_{A}=\left\vert \chi _{A}\right\vert $, the number
of observations sites in the sub-domain $\chi _{A}.$
\end{proof}

\subsubsection{Spatial stability of the MDA}

\begin{theorem}
Let $C_{H}$ be the spatial copula of the process $\left\{ Y\left( x\right)
\right\} $. Then, under the key assumption, the copula $C_{H}$ converge to a
spatial extremal copula $C_{G}$.
\end{theorem}

\begin{proof}
\textit{Let consider the following notation of component-wise vector of
spatio-temporal process. }%
\begin{eqnarray*}
Y_{t}^{\check{s}}\left( s\right) =Y\left( t,s\right) =\left\{ \left(
Y_{t_{1}}\left( s_{1}\right) ;...;Y_{t_{n}}\left( s_{n}\right) \right) ,s\in
S,t\in T\right\}
\end{eqnarray*}%
\textit{\ \ is the response vector at a given time t from a spatio-temporal
and max-stable model.}

\textit{So, under this notation a realisation }$y\left( t,s\right)
=y_{t}\left( s\right) $\textit{\ of }$\ Y_{t}\left( s\right) $\textit{\ is
obtained as}%
\begin{equation}
y_{i,t}\left( s\right) =\mu _{i,t}\left( s_{i}\right) +\frac{\sigma
_{i,t}\left( s_{i}\right) }{\xi _{i,t}\left( s_{i}\right) }\left[
s_{t}\left( s\right) ^{\xi _{t}^{\left( i\right) }\left( s\right) }-1\right] ~~for~~i=1,\dots,m.
\end{equation}
\textit{Equivalently, it comes that, for a given site s }$D_{N}=\left\{
s_{1},...,s_{N}\right\} \subset \mathbb{R}^{2}$\textit{\ }%
\begin{equation}
P\left( \frac{Y_{1}\left( s\right) -b_{n}^{\left( 1\right) }\left( s\right) 
}{a_{n}^{\left( 1\right) }\left( s\right) }\leq y_{1}\left( s\right) ;...;%
\frac{Y_{N}\left( s\right) -b_{n}^{\left( N\right) }\left( s_{n}\right) }{%
a_{n}^{\left( N\right) }\left( s_{n}\right) }\leq y_{N}\left( s\right)
\right) ^{n}=H\left( y_{1}\left( s\right) ;...;y_{N}\left( s\right) \right) .
\end{equation}%
For simplicity reasons, let denote, like in the paper $\left[ 6%
\right] $\textit{\ that }$Y\left( t,s\right) =Y_{t}^{\check{s}}$\textit{\
(which is different from }$Y_{t}^{s}$\textit{, the s-th power of }$Y_{t}).$%
Then, under this notational assumption the spatialized version of
the joint distribution function F of Y is given by $F_{t}^{\check{s}}$%
for given vector of realization $y_{t}^{\check{s}}=\left(
y_{t}^{\left( 1\right) }\left( s\right) ,...,y_{t}^{\left( m\right) }\left(
s\right) \right) $ such as
\begin{eqnarray*}
F_{t}^{\check{s}}\left( y_{1}\left( t,s\right) ,\dots,y_{m}\left( t,s\right)
\right) =F\left( y_{1}^{\check{s}_{1}}\left( t\right)
;\dots;y_{m}^{s_{m}}\left( t\right) \right) =F\left( y_{1}\left( t,s\right)
;\dots;y_{m}\left( t,s\right) \right) .
\end{eqnarray*}%
In the same vein, the spatio-temporal copula associated to the
distribution G via Sklar parametrization (1) will be denoted as $C_{t}^{%
\check{s}}=\left( C_{1,t}^{\check{s}};\dots;C_{m,t}^{\check{s}}\right) .$

The key assumption insures that the distribution of the process $\left\{
Y\left( x\right) \right\} $ lies in the domain of attraction of a
multivariate EVdistribution G. In particular, marginally there exist
appropriate spatial coefficients of normalization $\left\{ \sigma
_{n,i}\left( x_{i}\right) >0\right\} $ and $\left\{ \mu _{n,i}\left(
x_{i}\right) \in 
%TCIMACRO{\U{211d} }%
%BeginExpansion
\mathbb{R}
%EndExpansion
\right\} $ such as%
\begin{equation}
\lim_{n\rightarrow +\infty} H_{i}^{n}\left( \sigma
_{n,i}\left( x_{i}\right) y_{i}\left( x_{i}\right) +\mu _{n,i}\left(x_{n,i}\right) \right) = G_{i}\left( y_{i}\left( x_{i}\right) \right) ~~for ~~all ~~i=1,\dots,n.
\end{equation}

More generally, in one hand, applying (7) to the joint dependence structure,
it follows that 
\begin{eqnarray}
\left.
\begin{tabular}{l}
$\lim_{n\rightarrow +\infty} H^{n}\left( \sigma _{n,1}\left(
x_{1}\right) y_{1}\left( x_{1}\right) +\mu _{n,1}\left( x_{n,1}\right)
;\dots;\sigma _{n,s}\left( x_{s}\right) y_{s}\left( x_{s}\right) +\mu
_{n,s}\left( x_{n,s}\right) \right) $ \\ 
\\ 
$\ =G\left( y_{1}\left( x_{1}\right) ;\dots;y_{s}\left( x_{s}\right) \right)
=C_{G}\left( G\left( y_{1}\left( x_{1}\right) \right) ;\dots;G_{s}\left(
y_{s}\left( x_{s}\right) \right) \right) .$
\end{tabular}
\right\rbrace
\label{iit1}
\end{eqnarray}

On the other hand however, 
\begin{equation}
\left. 
\begin{tabular}{l}
$H^{n}\left( \sigma _{n,1}\left( x_{1}\right) y_{1}\left( x_{1}\right) +\mu
_{n,1}\left( x_{n,1}\right) ;...;\sigma _{n,s}\left( x_{s}\right)
y_{s}\left( x_{s}\right) +\mu _{n,s}\left( x_{n,s}\right) \right) $ \\ 
\\ 
$\ =C_{H}^{n}\left( H_{1}\left( \sigma _{n,1}\left( x_{1}\right) y_{1}\left(
x_{1}\right) +\mu _{n,1}\left( x_{1}\right) \right) ;...;H_{s}\left( \sigma
_{n,s}\left( x_{s}\right) y_{s}\left( x_{s}\right) +\mu _{n,s}\left(
x_{s}\right) \right) \right) .$%
\end{tabular}%
\right\rbrace
\end{equation}

Moreover, the copula C$_{H}$ verifies the property of max-stability given by
the relation (5).

Then, it results an asymptotical copula such as%
\begin{equation}
\left. 
\begin{tabular}{l}
$\lim_{n\rightarrow +\infty} H^{n}\left( \sigma _{n,1}\left(
x_{1}\right) y_{1}\left( x_{1}\right) +\mu _{n,1}\left( x_{n,1}\right)
;\dots;\sigma _{n,s}\left( x_{s}\right) y_{s}\left( x_{s}\right) +\mu
_{n,s}\left( x_{n,s}\right) \right) $ \\ 
\\ 
$\ =\lim_{n\rightarrow +\infty} C_{H}^{n}\left( H_{1}\left(
\sigma _{n,1}\left( x_{1}\right) y_{1}\left( x_{1}\right) +\mu _{n,1}\left(
x_{1}\right) \right) ;\dots;H_{s}\left( \sigma _{n,s}\left( x_{s}\right)
y_{s}\left( x_{s}\right) +\mu _{n,s}\left( x_{s}\right) \right) \right) $ \\ 
\\ 
$\ \ =C_{H}\left( \left( Gy_{1}\left( x_{1}\right) \right) ;\dots;G_{s}\left(
y_{s}\left( x_{s}\right) \right) \right) .$
\end{tabular}
\right\rbrace
\label{iit}
\end{equation}

Therefore, using simultaneously \ref{iit1} and \ref{iit} it follows that, for all
realization $y(x)$ of $\left\{ Y\left( x\right) \right\} $%
\begin{eqnarray*}
C_{G}\left( G\left( y_{1}\left( x_{1}\right) \right) ;...;G_{s}\left(
y_{s}\left( x_{s}\right) \right) \right) =C_{H}\left( \left( Gy_{1}\left(
x_{1}\right) \right) ;...;G_{s}\left( y_{s}\left( x_{s}\right) \right)
\right) .
\end{eqnarray*}%
Therefore, the uniqueness of the copula associated to the continuous
distribution H (Sklar, 1959) allows us to conclud that $C_{H}$ is
max-stable. Finally, the max-stability implies that $C_{H}$ is an extremal
copula.
\end{proof}

\section{Mains results}

The extremal coefficient is the natural dependence measures for extreme
value models which provides the magnitude of the asymptotic dependence of a
random field at two points of the domain.

\subsection{ Extremal dependence index and Copulas function}

The study of extreme value theory have been extended both to spatial and
multivariate contexts these last years. This section gives the relationship
between the extremal coefficient via copula.

\begin{theorem}
Let $\left\{ Z\left( s\right) ,s\in \mathbb{R}^{2}\right\} $ be stationary max-stable random process with Fréchet
marginal. Then, the extremal copula-based coefficient is given by:%
\begin{equation}
\theta (h)=\left\{ 
\begin{array}{cc}
u_{\beta }(z)\left[ \mu +\frac{\int_{0}^{1}F_{Z}^{-1}(u)dC_{h}(u,u)-\mu }{\Gamma (1-\xi )}\right] & if~\xi \neq 0 \\ 
&  \\ 
\exp \left\{ \frac{\int_{0}^{1}F_{Z}^{-1}(u)dC_{h}(u,u)-\mu }{\sigma }%
\right\} & if~\xi =0%
\end{array}%
\right. ,  \label{opor}
\end{equation}%
where 
\[
u_{\beta }(z)=\left\{ 
\begin{array}{cc}
\left[ 1+\xi \left( \frac{z-\mu }{\sigma }\right) \right] ^{1/\xi } & 
if~1+\xi \left( \frac{z-\mu }{\sigma }\right) >0 \\ 
\\
0 & if~1+\xi \left( \frac{z-\mu }{\sigma }\right) \leq 0%
\end{array}%
\ \right. ;for~all~z\in \mathbb{R},
\]%
and 
\[
\forall z>0,~~\Gamma (z)=\int_{0}^{+\infty }t^{z-1}e^{-t}dt.
\]
\end{theorem}

\begin{proof}
Let Z be a stationary random field of the second order of form parameter $\xi< 1$. The extremal
coefficient is given using the underlying madogram by:

\[
\theta (h)=\left\{ 
\begin{array}{cc}
u_{\beta }(z)\left[ \mu +\frac{M(h)}{\Gamma (1-\xi )}\right] & if~\xi \neq 0
\\ 
&  \\ 
\exp \left\{ \frac{M(h)}{\sigma }\right\} & if~\xi =0%
\end{array}%
\ \right. .
\]%
where $M_{h}$ is the semi-variogram given by: 
\begin{equation}
M(h)=\frac{E(|Z(x+h)-Z(x)|)}{2}.  \label{opo}
\end{equation}%
So, for all, x$\in 
%TCIMACRO{\U{211d} }%
%BeginExpansion
\mathbb{R}
%EndExpansion
^{2}$ and by taking into account the fact that 
\[
|Z(x+h)-Z(x)|=2\max [Z(x+h),Z(x)]-Z(x+h)-Z(x)
\]%
the relation (3.2) provides: 
\[
M(h)=\frac{E\left( 2\max [Z(x+h),Z(x)]-Z(x+h)-Z(x)\right) }{2}.
\]%
So, it follows that: 
\[
M(h)=E\left( \max [Z(x+h),Z(x)]\right) -\frac{1}{2}\left[E(Z(x+h))+E(Z(x))\right]
\]%
Then, for a stricly continous context, 
\begin{equation}
M(h)=E\left( \max [Z(x+h),Z(x)]\right) -\mu ,  \label{o'}
\end{equation}%
where $\mu =E(Z(x+h))=E(Z(x)),$ is the means of $Z(.)$ is stationary in the
second order. 
\[
E\left( \max [Z(x+h),Z(x)]\right) =\int_{-\infty }^{+\infty }zdC_{h}\left(
F_{Z}(z),F_{Z}(z)\right) .
\]%
Which gives 
\begin{equation}
E\left( \max [Z(x+h),Z(x)]\right) =\int_{0}^{1}F_{Z}^{-1}(u)dC_{h}(u,u).
\label{t}
\end{equation}%
Then, using the formula (\ref{t}) in (\ref{o'}), one obtain 
\begin{equation}
M(h)=\int_{0}^{1}F_{Z}^{-1}(u)dC_{h}(u,u)-\mu .  \label{oo'}
\end{equation}%
So by using the relation (\ref{oo'}) in the expression of the coefficient
extremal we get 
\[
\theta (h)=\left\{ 
\begin{array}{cc}
u_{\beta }(z)\left[ \mu +\frac{\int_{0}^{1}F_{Z}^{-1}(u)dC_{h}(u,u)-\mu }{%
\Gamma (1-\xi )}\right] & if~\xi \neq 0 \\ 
&  \\ 
\exp \left\{ \frac{\int_{0}^{1}F_{Z}^{-1}(u)dC_{h}(u,u)-\mu }{\sigma }%
\right\} & if~\xi =0%
\end{array}%
\ \right. .
\]%
Finaly, it yields the relation (\ref{opor}) as disserted.\\
\end{proof}

Let $Z$ be a max-stable random field. The extremal coefficient and the
copula function are related differently depending on the marginal
distribution of the Z process. \newline

\begin{proposition}
Let Z be a spatial domaine distributed according a stationary max-stable model
G of with either or Gumbel or Weibull univariate marginal then, the extremal
coefficient is given by:%
\begin{equation}
\theta (h)=\left\{ 
\begin{array}{cc}
\frac{1}{1-G\left( F_{Z},C_{h},u\right)+\mu } & of~standard~Weibull, \\ 
&  \\ 
\exp \left( G\left( F_{Z},C_{h},u\right) -\mu \right) & ~of~standard~Gumbel%
\end{array}%
\ \right. .  \label{tafo}
\end{equation}%
where $G\left( F_{Z},C_{h},u\right) =\int_{0}^{1}F_{Z}^{-1}(u)dC_{h}(u,u).$
\end{proposition}

\begin{proof}
Dealing with the case where the margins of Z are distributed according the
Weibull model, it is well known that the extremal coefficient and the
madogram are associated by the relation $\theta (h)=\frac{1}{1-M(h)}.$ So,
using (\ref{oo'}) in this relation, it comes, under the existence, that $%
\theta (h)=\frac{1}{1-\int_{0}^{1}F_{Z}^{-1}(u)dC_{h}(u,u)+\mu }.$ Hence
the first result of (\ref{tafo}). \newline
\medskip

Similarly, if the margins of Z are Brown-Resnick model (see [12]), then $%
\theta (h)=\exp (M(h))$. So, using (\ref{oo'}) in this relationship, it
comes back that 
\[
\theta (h)=\exp \left( \int_{0}^{1}F_{Z}^{-1}(u)dC_{h}(u,u)-\mu \right) .
\]%
Hence the last result of (\ref{tafo})
\end{proof}

\begin{figure}[hbtp!!]
\includegraphics[scale=0.5]{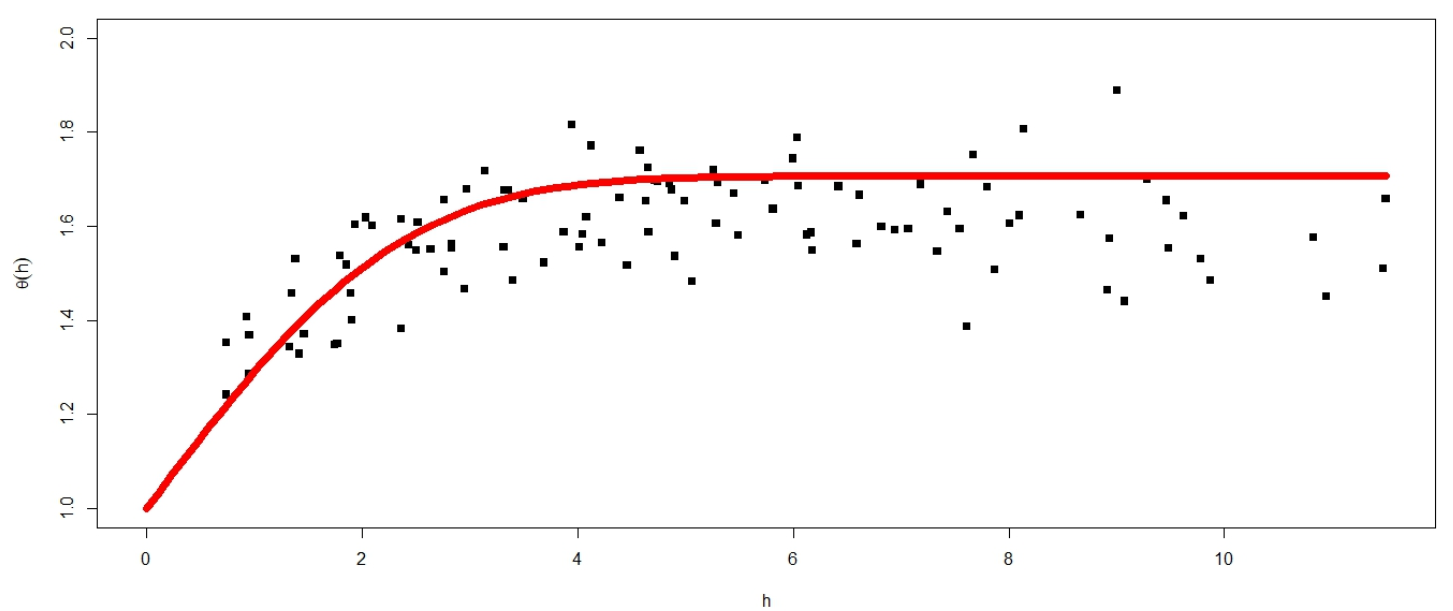} 
\caption{Graph of Extremal coefficient for a Brown-Resnick process.\protect
\ref{fig2}.}
\label{fig2}
\end{figure}

The following section allows us to construct a model of the extremogram and
the cross-extremogram via copula function to determine the distributional
dependence of the random variables of the random field Z depending on the
inter-site distance.

\subsection{Sampling extremogram with Copulas}

In this subsection, we model the extremogram function using a copula
function for all $A\subset \mathbb{R}_{\ast }^{+}$ and $a\in A$. We obtain
the following result see(\cite{martin};~\cite{muneya};~\cite{davis}).

\begin{theorem}
Consider $F_{i,Z}$ the distribution function of the random variable $Z_{i}$
and $U_{i}$ the uniform transformation of $F_{i,Z}$. Then, a copula-based
extremogram is given, for all $x_{i},x_{j}\in \mathbb{R}^{d}$, by: 
\[
\rho _{AA}(h_{ij})=\rho _{(a,+\infty )}(h_{ij})=2-\lim_{u\rightarrow 1^{-}}%
\frac{1-C_{h_{ij}}(u,u)}{1-u},
\]%
where $h_{ij}$ is the separating distance between $x_{i},x_{j}.$
\end{theorem}

\begin{proof}
It is well known that $\rho_{AA}(h_{ij})=\lim\limits_{z\rightarrow
\infty}P\left(\frac{Z(x_{j})}{z}\in A/\frac{Z(x_{i})}{z}\in A\right)$. Such
as: $A=(a,+\infty)$, this expression can be written as, 
\[
\rho_{AA}(h_{ij}) = \lim\limits_{z\rightarrow +\infty}\frac{P\left(\frac{%
Z(x_{j})}{z}> a,\frac{Z(x_{i})}{z}>a\right)}{P(\frac{Z(x_{i})}{z}>a)}.
\]
Then, it is easy to show that, 
\[
\rho_{AA}(h_{ij}) = \lim\limits_{z\rightarrow +\infty}\frac{%
P\left(Z(x_{j})> az,Z(x_{i})>az\right)}{P(Z(x_{i})>az)} ,
\]
Z being a stationary random field. Under the assumption that $%
F_{i}(az)=F_{j}(az)=u$. \medskip

Then, it follows that, 
\[
\rho _{AA}(h_{ij})=\lim\limits_{u\rightarrow 1^{-}}\frac{P(U_{j}>u,U_{i}>u)%
}{P(U_{i}>u)}.
\]%
Nevertheless, using the survival copula, when have: 
\[
P(U_{j}>u,U_{i}>u)=1-u-u+C_{h_{ij}}(u,u).
\]%
Therefore, 
\[
\rho _{AA}(h_{ij})=\lim\limits_{u\rightarrow 1^{-}}\frac{%
1-2u+C_{h_{ij}}(u,u)}{1-u}.
\]%
Then, based on a result of Cooley \& al. [$7$], it follows that: 
\[
\rho _{AA}(h_{ij})=\lim\limits_{u\rightarrow 1^{-}}\left[ 2-\frac{%
1-C_{h_{ij}}(u,u)}{1-u}\right] .
\]%
So, as disserted 
\[
\rho _{AA}(h_{ij})=\rho _{(a,+\infty )}(h)=2-\lim_{u\rightarrow 1^{-}}\frac{%
1-C_{h_{ij}}(u,u)}{1-u}.
\]%
In the particular case where $a=1$, that is $F_{i}(az)=F_{i}(z)$ the
extremogram merges with the upper tail dependence measure. So, 
\[
\rho _{AA}(h_{ij})=\lim_{u\rightarrow 1^{-}}2-\frac{1-C_{h_{ij}}(u,u)}{1-u}%
=2-\lim_{u\rightarrow 1^{-}}\frac{1-C_{h_{ij}}(u,u)}{1-u}=\chi (h_{ij}).
\]%
For the particular case where $a=1$. Moreover If $\rho _{(1,+\infty
)}(h_{ij})=0$, then the random variables $Z_{i}$ and $Z_{j}$ are
asymptotically independent \medskip 
\end{proof}

In a second case, considering that $A\subset \mathbb{R}_{\ast }^{-}$ and $%
a\in A$, we obtain next relation of the extremogram via the underlying
copula. In particular, if $\chi _{h}$\ is reduced to a single site $x$, the
law of $Y^{\ast }$\ is either the Frechet distribution, the Gumbel or the
Weibull distribution. \medskip

The following result provides a copula-based extension of the extremogram of
the process.

\begin{proposition}
The extremogram $\rho_{AA}$ and the copula function $C_{h_{ij}}$ are linked by the relation: 
\begin{equation}
\rho _{AA}(h_{ij})=\lim_{u\rightarrow 0^{+}}\frac{C_{h_{ij}}(u,u)}{u},~u\in[
0,1].  \label{oo}
\end{equation}
\end{proposition}

\begin{proof}
It is well known that $\rho _{AA}(h_{ij})=\lim\limits_{z\rightarrow -\infty
}P\left( \frac{Z(x_{j})}{z}\in A/\frac{Z(x_{i})}{z}\in A\right) $. \newline
Since $A=(-\infty ,a)$, it follows that: 
\[
\rho _{AA}(h_{ij})=\lim\limits_{z\rightarrow -\infty }P\left( \frac{Z(x_{j})%
}{z}\leq a\big/\frac{Z(x_{i})}{z}\leq a\right) .
\]%
Then, 
\[
\rho _{AA}(h_{ij})=\lim\limits_{z\rightarrow -\infty }P\left( Z(x_{j})\leq
az/Z(x_{i})\leq az\right) .
\]%
Thus, 
\[
\rho _{AA}(h_{ij})=\lim\limits_{u\rightarrow 0^{+}}P(U_{j}\leq u/U_{i}\leq
u).
\]%
Therefore, 
\[
\rho _{AA}(h_{ij})=\lim\limits_{u\rightarrow 0^{+}}\frac{P(U_{j}\leq
u,U_{i}\leq u)}{P(U_{i}\leq u)}.
\]
Hence the result (\ref{oo}) as disserted. \\
\end{proof}
\begin{figure}[hbtp!!]
\caption{Graph of theoretical Extremogram .}
\label{fig1}
\includegraphics[scale=0.5]{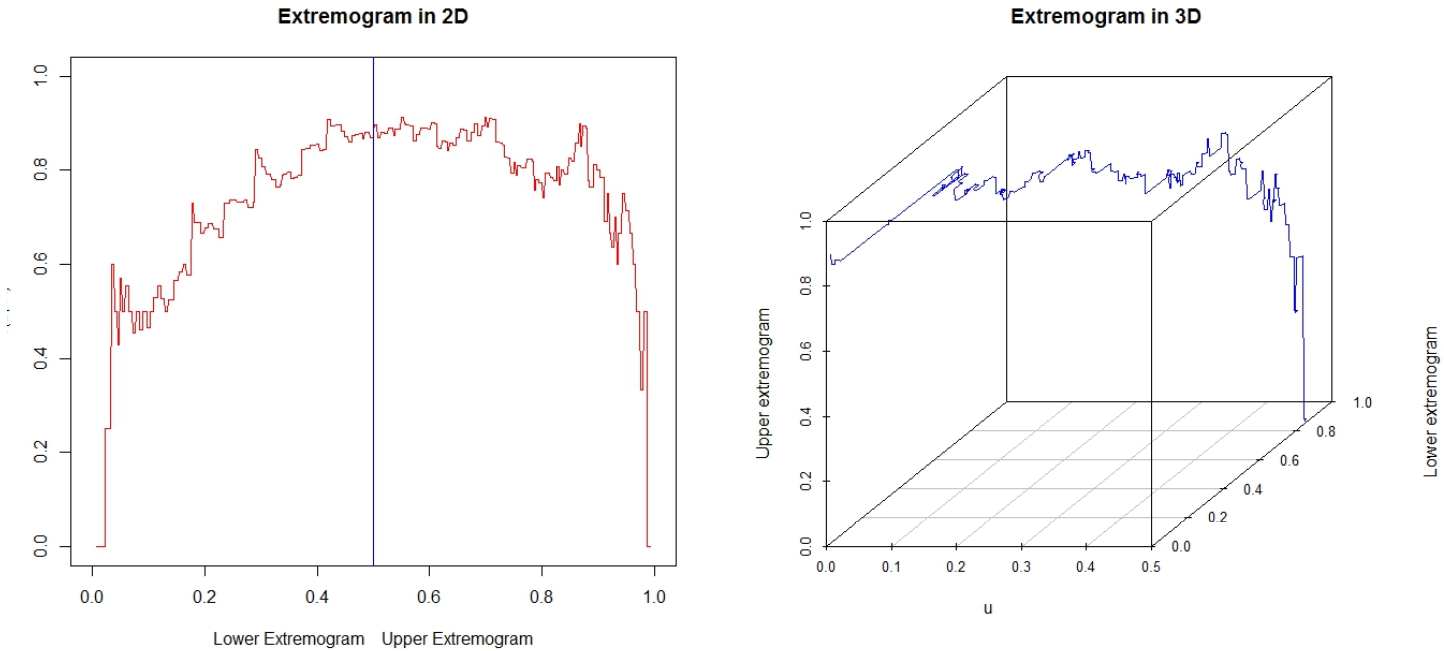}\\
{\small {\
This figures gives the representation in $2D$ and $3D$ for $A=(-\infty,1)$
and $A=(1,+\infty)$. We denote by upper extremogram for $A=(1,+\infty)$ and
lower extremogram for $A=(-\infty,1).$} above the figure.}
\end{figure}

The following subsection gives a relation between the cross-extremogram and
the copula function.

\subsection{Cross-extremogram sampling with Spatial copulas}

\indent The following result provides a characterization of the cross
extremogram in a copula contex, for two given sites $s_{i}$ and $s_{j}$.

\begin{theorem}
For two given sites $s_{i}$ and $s_{j}$ separated by $h_{ij}$, the extremal
coefficient it given by: 
\begin{equation}
\rho _{AB}(h_{ij})=1-\lim_{(u_{1i},u_{2j})\rightarrow (1^{-},1^{-})}\frac{u_{2j}-C_{h_{ij}}(u_{1i},u_{2j})}{1-u_{1i}}.  \label{ismo}
\end{equation}
\end{theorem}

\begin{proof}
Let $F_{ji}(Z(x_{i}))=U_{ji}$ be the univariate distribution functions
obtained by integral transforms to the variables $Z_{j}(x_{i})$ with $%
x_{j}-x_{i}=h_{ij},1\leq i,j\leq n,i\neq j$. It is well known that 
\[
\rho _{AB}(h_{ij})=\lim_{z\rightarrow \infty }P\left( Z_{2}(s_{j})\in
zB/Z_{1}(s_{i})\in zA\right) .
\]%
Since $A=(a,\infty )$ and $B=(b,\infty )$, it follows that 
\[
\rho _{AB}(h_{ij})=\lim_{z\rightarrow \infty }P\left( Z_{2}(s_{j})\in
zB/Z_{1}(s_{i})\in zA\right) .
\]%
It follow that, 
\[
\rho _{AB}(h_{ij})=\lim_{z\rightarrow \infty }P\left( Z_{2}(s_{j})\geq
bz/Z_{1}(s_{i})\geq az\right) .
\]%
So, 
\[
\rho _{AB}(h_{ij})=\lim_{z\rightarrow \infty }\frac{P\left(
Z_{2}(s_{j})\geq bz,Z_{1}(s_{i})\geq az\right) }{P\left( Z_{1}(s_{i})\geq
az\right) }.
\]%
\[
\rho _{AB}(h_{ij})=\lim_{z\rightarrow \infty }\frac{\hat{H}_{h_{ij}}(bz,az)%
}{P\left( Z_{1}(s_{i})\geq az\right) }
\]%
with $\hat{H}_{h_{ij}}(bz,az)$ the survival function of the variables $%
Z_{2}(s_{j})$ and $Z_{1}(s_{i}).$\newline
\newline
Moreover, if $C_{h_{ij}}$ is the jointed copula underlying the distribution
of $Z_{1}(s_{i})$ and $Z_{2}(s_{j})$, then, it follows that: 
\[
\hat{H}%
_{h_{ij}}(bz,az)=1-F_{1i}(az)-F_{2j}(bz)+H_{h_{ij}}(bz,az)=1-F_{1i}(az)-F_{2j}(bz)+C_{h_{ij}}(F_{2j}(bz),F_{1i}(az)),
\]%
\indent Likewise 
\[
P\left( Z_{1}(s_{i})\geq az\right) =1-P\left( Z_{1}(s_{i})<az\right)
=1-F_{1i}(az).
\]%
By replacing these two last relations in (\ref{o'}), we obtain the following
result: 
\begin{equation}
\rho _{AB}(h_{ij})=\lim_{z\rightarrow \infty }\frac{%
1-F_{1i}(az)-F_{2j}(bz)+C_{h_{ij}}(F_{1i}(az),F_{2j}(bz))}{1-F_{1i}(az)}.
\label{pif}
\end{equation}%
Let us consider $u_{1i}=F_{1i}(az)$ et $u_{2j}=F_{2j}(bz)$. When $%
z\longrightarrow +\infty $ then $u_{1i}\longrightarrow 1^{-}$ and $%
u_{2j}\longrightarrow 1^{-}.$ By using these transformations in the relation
(\ref{pif}), it follows that: 
\[
\rho _{AB}(h_{ij})=\lim_{(u_{1i},u_{2j})\rightarrow (1^{-},1^{-})}\frac{%
1-u_{1i}-u_{2j}+C_{h_{ij}}(u_{1i},u_{2j})}{1-u_{1i}}.
\]%
So \[
\rho _{AB}(h_{ij})=1-\lim_{(u_{1i},u_{2j})\rightarrow (1^{-},1^{-})}\frac{u_{2j}-C_{h_{ij}}(u_{1i},u_{2j})}{1-u_{1i}}.
\]%
Hence the result (\ref{ismo}) as disserted
\end{proof}

The following results provides an asymptotic statement.

\begin{proposition}
Consider $F_{ij}(az)=u_{ij}\in \lbrack 0,1]$ the distribution function of
the variable $Z_{j}(x_{i})$. If $z\longmapsto +\infty $, then $%
u_{ij}\longmapsto 1^{-}$ The relation (\ref{tifa}) is written according to
the copula by the relation: 
\begin{equation}
\left( 
\begin{array}{c}
\rho _{AA}^{11}(h_{ij}) \\ 
\\ 
\\ 
\rho _{BB}^{22}(h_{ij}) \\ 
\\ 
\\ 
\rho _{AB}^{12}(h_{ij}) \\ 
\\ 
\\ 
\rho _{BA}^{21}(h_{ij})%
\end{array}%
\right) =\left( 
\begin{array}{c}
1-\displaystyle\lim_{u_{11}\rightarrow 1^{-}}\frac{u_{11}-C_{h_{ij}}\left(
u_{11},u_{11}\right) }{1-u_{11}} \\ 
\\ 
1-\displaystyle\lim_{u_{22}\rightarrow 1^{-}}\frac{u_{22}-C_{h_{ij}}\left(
u_{22},u_{22}\right) }{1-u_{22}} \\ 
\\ 
1-\displaystyle\lim_{(u_{11},u_{22})\rightarrow (1^{-},1^{-})}\frac{u_{22}-C_{h_{ij}}\left( u_{22},u_{11}\right) }{1-u_{11}} \\ 
\\ 
1-\displaystyle\lim_{(u_{11},u_{22})\rightarrow (1^{-},1^{-})}\frac{u_{11}-C_{h_{ij}}\left( u_{11},u_{22}\right) }{1-u_{22}}%
\end{array}%
\right)  \label{papito}
\end{equation}
\end{proposition}

\begin{proof}
In matrix form, the extremogram and the crossed extremogram can be written,
(see Muneya et al. [$18$]), for all $x,x+h\in \mathbb{R}^{d}$, such as, 
\begin{equation}
\left( 
\begin{array}{c}
\rho _{AA}^{11}(h) \\ 
\\ 
\rho _{BB}^{22}(h) \\ 
\\ 
\rho _{AB}^{12}(h) \\ 
\\ 
\rho _{BA}^{21}(h)%
\end{array}%
\right) =\lim_{z\rightarrow +\infty }\left( 
\begin{array}{ccc}
P\left( Z_{1}(x+h)\in zA/Z_{1}(x)\in zA\right) &  &  \\ 
&  &  \\ 
P\left( Z_{2}(x+h)\in zB/Z_{2}(x)\in zB\right) &  &  \\ 
&  &  \\ 
P\left( Z_{2}(x+h)\in zB/Z_{1}(x)\in zA\right) &  &  \\ 
&  &  \\ 
P\left( Z_{1}(x+h)\in zA/Z_{2}(x)\in zB\right) &  & 
\end{array}%
\right)  \label{tifa}
\end{equation}%
Consider $B=A=(a,+\infty )$. Z being stationary, let $u_{1i}=u_{1j}=u_{11}$.
With these transformations the relation (\ref{ismo}) is written in the form, 
\[
\rho _{AA}^{11}(h_{ij})=\displaystyle\lim_{u_{11}\rightarrow 1^{-}}\frac{%
1-2u_{11}+C_{h_{ij}}\left( u_{11},u_{11}\right) }{1-u_{11}}.
\]%
Hence the first expression of (\ref{papito}).\newline
\newline
In the same way, let us consider that $A=B=(b,+\infty )$ and $%
u_{2j}=u_{2i}=u_{22}$, the relation (\ref{ismo}) is written in the form, 
\[
\rho _{BB}^{22}(h_{ij})=\displaystyle\lim_{u_{22}\rightarrow 1^{-}}\frac{%
1-2u_{22}+C_{h_{ij}}\left( u_{22},u_{22}\right) }{1-u_{22}}.
\]%
Hence the second expression of (\ref{papito}).\newline
\newline
Similarly for $A=(a,+\infty )$ and $B=(b,+\infty )$, let $%
u_{1i}=u_{1j}=u_{11}$ and $u_{2j}=u_{2i}=u_{22}$. The relation (\ref{ismo})
is written in the form, 
\[
\rho _{AB}^{12}(h_{ij})=\displaystyle\lim_{(u_{11},u_{22})\rightarrow
(1^{-},1^{-})}\frac{1-u_{22}-u_{11}+C_{h_{ij}}\left( u_{22},u_{11}\right) }{%
1-u_{11}}.
\]%
By swapping A and B, $u_{11}$ and $u_{22}$ will change location. So this new
relationship is still written in the form, 
\[
\rho _{BA}^{21}(h_{ij})=\displaystyle\lim_{(u_{22},u_{11})\rightarrow
(1^{-},1^{-})}\frac{1-u_{11}-u_{22}+C_{h_{ij}}\left( u_{11},u_{22}\right) }{%
1-u_{22}}.
\]%
Hence the third and fourth expressions of (\ref{papito}).
\end{proof}

The following section is used to characterize the asymptotic dependence of
extremes through the extremogram.

\subsection{Asymptotic dependence and extremogram model}

Consider a random variable $T(x)$ of a spatial process $T=\left\{ T(x),x\in 
\mathbb{R}^{d}\right\} $ of standardized marginalized $F_{T}(T(x))$ .

\begin{theorem}
Let $Z=\left\{ Z(x),x\in \mathbb{R}^{d}\right\} $ be a spatial stationary
process such that $Z(.)=\frac{-1}{\log (F_{T}(T(.)))}$ . The marginal
distribution of $Z$ are Fréchet standard marginal. The
extremogram of random field Z(.) in two sites $x,x+h\in \mathbb{R}^{d}$ is
define such as, 
\begin{equation}
\rho _{AA}(h)=\mathcal{L}_{h}(u)(u)^{1-\frac{1}{\eta _{h}}};  \label{iff}
\end{equation}%
where $\eta (h)\in (0,1]$ is the tail dependence coefficient, $A=(a;\infty
)~~with~~a\in (0,1]$ and $\mathcal{L}(.)$ a slowly varying function.\newline
\end{theorem}

Before giving the proof of the above theorem, let's note that, even in a
spatial study, there no loss of generality in dealing with Fr\'{e}chet
marginal, for any continuous function f, the transformation $f\left(
Y_{i}\left( x\right) \right) =\frac{-1}{\log \left( Y_{i}\left( x\right)
\right) }$\ gives approximatively this distribution.\medskip\ Indeed, the
parameters of the GEV in (2) as smooth function of the explanatory variables
(longitude, altitude, elevation etc.) such as: 
\[
Y\left( x\right) =\mu \left( x\right) +\frac{\sigma \left( x\right) }{\xi
\left( x\right) }\left[ Z\left( x\right) ^{\xi \left( x\right) }-1\right] 
~~ where~~ Z\left( x\right) \sim ~Unit-Fréchet,
\]%
for some partially correlation. That needs to model both spatial behaviour
of marginal parameters and spatial joint dependence.\bigskip

\begin{proof}
Considering $A=(a;\infty ),a\in (0;1]$, the extremogram is written 
\[
\rho _{AA}(h)=\lim\limits_{z\rightarrow +\infty }\frac{P(Z(x+h)>az,Z(x)>az)%
}{P(Z(x)>az)}.
\]

According to Ledford and Tawn $[16]$, when z tends towards infinity, 
\begin{equation}
P(Z_{2}>r,Z_{1}>r)\sim \mathcal{L}(r)(r)^{1-\frac{1}{%
\eta }}.  \label{brav}
\end{equation}%
Using (\ref{brav}), for any spatial process Z at two sites $x$ and $x+h$ when z tends towards infinity, we
can write 
\begin{equation}
P(Z(x+h)>az/Z(x)>az)\displaystyle\sim \mathcal{L}_{h}(az)(az)^{1-%
\frac{1}{\eta _{h}}}.
\end{equation}%
Thus, let $F(Z(x))$ be the distribution function of $Z(x)$ and $F(Z(x+h))$
the distribution function of $Z(x+h)$. According to the above, when z tends towards infinity, it follows
that: 
\begin{equation}
P(F(Z(x+h))>F(az)/F(Z(x))>F(az))\sim \mathcal{L}%
_{h}(az)(az)^{1-\frac{1}{\eta _{h}}}.  \label{brav1}
\end{equation}
Considering $U=F(Z(x))$, $V=F(Z(x+h))$ and $u=F(az)$, it follows that: 
\begin{equation}
P(V>u/U>u)\sim \mathcal{L}_{h}(u)(u)^{1-\frac{1}{\eta _{h}}},
\label{brav3}
\end{equation}
when z tends towards infinity.\\
Using (\ref{brav3}) in the expression of the extremogram, it follows that: 
\[
\rho _{AA}(h)=\mathcal{L}_{h}(u)(u)^{1-\frac{1}{\eta _{h}}}
\]%
Hence (\ref{iff}) as disserted
\end{proof}

Ancona and Tawn $[2]$ proposed a measure of extreme dependence called
extreme variogram. This measure of dependence is expressed as a function of
the dependence of tail by the relation: 
\begin{equation}
\gamma _{E}(h)=2(1-\eta (h)).  \label{if}
\end{equation}%
Thus, the extremogram is modeled according to the extreme variogram by the
following result.

\begin{corollary}
Let $\gamma _{E}(h)$ be the extreme variogram of two stationary random
variables. The extremogram is linked to the extreme variogram by the
relation:

\begin{equation}
\rho _{AA}(h)=\mathcal{L}_{h}(u)(u)^{-\frac{\gamma _{E}(h)}{2-\gamma _{E}(h)%
}};  \label{fi}
\end{equation}%
with $\gamma _{E}(h)\in \lbrack 0;2)$
\end{corollary}

\begin{proof}
From the relation (\ref{if}), we can say that $\eta (h)=1-\frac{1}{2}\gamma
_{E}(h).$. Using this relation in (\ref{iff}), it follows that: 
\[
\rho _{AA}(h)=\mathcal{L}_{h}(u)(u)^{1-\frac{2}{2-\gamma _{E}(h)}}.
\]%
Hence the expression, 
\[
\rho _{AA}(h)=\mathcal{L}_{h}(u)(u)^{\frac{-\gamma _{E}(h)}{2-\gamma _{E}(h)%
}}.
\]%
\newline
\end{proof}

In the following, we estimate the extremogram using the relation (\ref{iff}%
). In this relationship, estimation of the extremogram requires estimation
of the slowly varying function and the tail dependence coefficient. The
following result gives the estimate of the extremogram.

\begin{proposition}
Consider two spatial random variables $Z(x)~and~Z(x+h);~x,~x+h\in \mathbb{R}%
^{d}$ of respective marginal distribution function $%
F_{Z}(Z(x))~and~F_{Z}(Z(x+h))$ . Let W$(.)$ considering 
\[
W(h)=\min \left\{ \frac{-1}{\log \left( F_{Z}(Z(x)\right) };\frac{-1}{\log
\left( F_{Z}(Z(x+h)\right) }\right\} ,
\]%
the estimated extremogram is written, for a fixed threshold $u_{h}$, in the
form, 
\begin{equation}
\hat{\rho}_{AA}(h)=\hat{c}_{h}(u)(u)^{1-\frac{1}{\hat{\eta}(h)}}.
\end{equation}%
Where 
\[
\hat{c}_{h}(u)=\frac{n_{u_{h}}}{n}u_{h}^{\frac{1}{\hat{\eta}(h)}}~~~~;~~~~%
\hat{\eta}(h)=\frac{1}{n_{u_{h}}}\sum_{k=1}^{n_{u_{h}}}\log \left\{ \frac{%
\hat{w}_{k}(h)-u_{h}}{u_{h}}\right\} ,
\]%
with $\hat{w}_{k}(h),k=1,\dots ,n_{u_{h}}$ are the observations $\hat{W}(h)$
exceeding the threshold $u_{h}$.
\end{proposition}

\begin{proof}
The extremogram is expressed by the relation, 
\[
\rho_{AA}(h)=\lim\limits_{z\rightarrow +\infty}\frac{P(Z(x+h)> z,Z(x)>z)}{%
P(Z(x)>z)}=\lim\limits_{z\rightarrow +\infty}\frac{P(W(h)>z)}{P(Z(x)>z)}.
\]
Ledford $([15];~ [16])$ proposed to consider $\mathcal{L}_{h}(u)$ as
constant that is,$\mathcal{L}_{h}(u)=c_{h}$ for all values z exceeding the
threshold $u_{h}$. Using the observations of the independent replications of
the spatial process approximate independent observations on $\hat{W}(h)$ are
obtained, where $\hat{W}(h)$ is the approximation to the variable $W(h)$.
From model (\ref{brav1}) and n independent observations, the log-likelihood
is 
\[
l(c_{h},\eta_{h})=(n-n_{u_{h}})\log \left(1-\frac{c_{h}}{u_{h}^{1/\eta_{h}}}%
\right)+n_{u_{h}}\log\left(\frac{c_{h}}{\eta_{h}}-c_{h}\right)-\frac{1}{%
\eta_{h}}\sum_{i=1}^{n_{u_{h}}}\hat{w}_{i}(h),
\]
where $\left\lbrace \hat{w}_{i}(h)\right\rbrace, i=1,\dots,n_{u_{h}}$ are
the observations of $\hat{W}(h)$ above the threshold $u_{h}$. Using the
maximum likelihood method, the estimate of $c_{h}$ is written, 
\[
c_{h}=\frac{n_{u_{h}}}{n}u_{h}^{\frac{1}{\hat{\eta}(h)}},
\]
and using the Hill estimator method, the estimate of $\eta_{h}$ is written, 
\[
\hat{\eta}(h)=\frac{1}{n_{u_{h}}}\sum_{k=1}^{n_{u_{h}}}\log \left\lbrace 
\frac{\hat{w}_{k}(h)-u_{h}}{u_{h}}\right\rbrace.
\]
Where $\hat{w}_{k}(h),k=1,\dots,n_{u_{h}} $ are the observations $\hat{W}(h)$
exceeding the threshold $u_{h}$. Hence the result, 
\begin{equation}
\hat{\rho}_{AA}(h)=\hat{c}_{h}(u)(u)^{1-\frac{1}{\hat{\eta}(h)}}.
\end{equation}
\end{proof}
\section{Conclusion and Discussion}

In this study, we have been modeling some technical tools of spatial
prediction within a copula-based space. \ Thus, the extremal coefficient and
the extremogram have been expressed via the underlying copulas. These
results are important insofar as we want to determine the inter-site
distribution dependence of a definite area.

The results of this paper make it possible to find a relation between the
extremal coefficient and the extremogram using \ the copula function. These
new model are very crucial since the copula is a parametrization of \ nomber
of variables which do not deal with the marginal distribution. Hence, they
allow not only to determine the distributional dependence of spatial or
temporal extremes, but also, and above all, the conditional distributional
dependence between these extremes in various observation sites.


%\chapter*{Bibliographie}
%\markboth{Bibliographie }{}
%\addcontentsline{toc}{chapter}{Bibliographie}

\end{document}